\documentclass[conference,a4paper]{IEEEtran}

\usepackage{amsthm}
\newtheorem{theorem}{Theorem}
\newtheorem{lemma}{Lemma}
\newtheorem{remark}{Remark}
\newtheorem{definition}{Definition}
\newtheorem{corollary}{Corollary}
\newtheorem{conjecture}{Conjecture}

\newcommand{\EE}{\mathbb{E}}
\usepackage{cite}
\usepackage{verbatim}

\usepackage[cmex10]{amsmath}
\usepackage{amssymb}

\usepackage{subfig}
\usepackage{graphicx}
\usepackage{xcolor}
\usepackage{rotating}

\usepackage{fixltx2e}
\usepackage{algorithmic,algorithm}
\usepackage{pseudocode}
\usepackage{mathrsfs}

\IEEEoverridecommandlockouts

\renewcommand{\Pr}{\mathscr{Pr}}
\begin{document}

\title{Which Boolean Functions are Most Informative?}

\author{\IEEEauthorblockN{Gowtham R. Kumar and Thomas A. Courtade}
\IEEEauthorblockA{Department of Electrical Engineering\\
Stanford University\\
Stanford, California, USA\\
Email: \{gowthamr, courtade\}@stanford.edu}}
\maketitle

\begin{abstract}
We introduce a simply stated conjecture regarding the maximum mutual information a Boolean function can reveal about noisy inputs.  Specifically, let $X^n$ be i.i.d.\ Bernoulli($1/2$), and let $Y^n$ be the result of passing $X^n$ through a memoryless binary symmetric channel with crossover probability $\alpha$.  For any Boolean function $b:\{0,1\}^n\rightarrow \{0,1\}$, we conjecture that $I(b(X^n);Y^n)\leq 1-H(\alpha)$. While the conjecture remains open, we provide substantial evidence supporting its validity.
\end{abstract}

\section{Introduction}\label{sec:Intro}
This paper is inspired by the following conjecture:
\begin{conjecture}\label{conj:Ib_Yn}
Let $X^n$ be i.i.d.\ Bernoulli(1/2), and let $Y^n$ be the result of passing $X^n$ through a memoryless binary symmetric channel with crossover probability $\alpha$.  For any Boolean function $b: \{0,1\}^n\rightarrow \{0,1\}$, we have
\begin{align}
I(b(X^n);Y^n)\leq 1-H(\alpha).\label{eqn:conjUB}
\end{align}
\end{conjecture}

At first sight, Conjecture \ref{conj:Ib_Yn}  might appear suitable as a homework exercise for a first course on information theory.  However, over the course of this paper, we hope to convince the reader that the conjecture is much deeper than it appears.  Despite its apparent simplicity, standard information-theoretical manipulations appear incapable of establishing \eqref{eqn:conjUB}.

To the present authors, Conjecture \ref{conj:Ib_Yn} represents the simplest, nontrivial
embodiment of Boolean functions in an information-theoretic context. In words, Conjecture \ref{conj:Ib_Yn} asks: ``\emph{What is the most significant bit that $X^n$ can provide  about $Y^n$?}"

Despite their fundamental roles in   computer science and digital computation, Boolean functions have received relatively little attention from the information theory community.  The recent work \cite{bib:Klotz2012} is perhaps most relevant to our Conjecture \ref{conj:Ib_Yn} and provides compelling motivation for its study.  In \cite{bib:Klotz2012}, the authors  prove that for $n$ and $\Pr\{b(X^n)=0\}\geq 1/2$ fixed, $I(b(X^n);X_1)$ is maximized by functions $b$ which satisfy $b(X^n)=0$ whenever $X_1=0$ (i.e., when $b$ is \emph{canalizing} in $X_1$).  The motivation for considering this problem comes from computational biology, where Boolean networks are used to model dependencies in various regulatory networks.  We encourage the reader to refer to \cite{bib:Klotz2012, bib:Samal2008} and the references therein for further information.

We remark that we have proven   the weaker inequality
\begin{align}
\sum_{i=1}^n I(b(X^n);Y_i) \leq 1-H(\alpha) \label{eqn:weakConjectureSum}
\end{align}
using Fourier-analytic techniques similar to those employed in \cite{bib:Klotz2012}.  However, this Fourier-analytic approach appears incapable of  establishing the stronger statement of Conjecture \ref{conj:Ib_Yn}.  We omit the proof of \eqref{eqn:weakConjectureSum} in this paper due to space constraints.

Conjecture \ref{conj:Ib_Yn} is also  related  to the \emph{Information Bottleneck Method} \cite{bib:Tishby1999}, which attempts to solve the optimization problem
\begin{align}
\min_{p(u|x^n)} I(X^n;U)-\lambda I(Y^n;U).
\end{align}
For a given $\lambda>0$, the optimizing $U$ is purportedly the best tradeoff between the accuracy of describing $Y^n$ and the descriptive complexity of $U$.  In our setting, $b(X^n)$ plays the role of $U$, and we constrain the descriptive complexity to be at most one bit.  It is relatively easy to show that randomized Boolean functions do not yield a higher mutual information.  Thus, expressing Conjecture \ref{conj:Ib_Yn} in terms of deterministic Boolean functions comes without loss of generality.

A more concrete example comes in the context of gambling.  To this end, suppose $Y^n$ is a  simple model for a market of $n$ stocks, where each stock doubles in value or goes bankrupt with probability 1/2, independent of all other stocks.  If an oracle has access to side information $X^n$, and we are allowed to ask  \emph{one} yes/no question of the oracle, which question should we ask to maximize the  rate at which our wealth grows?  The validity of Conjecture 1 would imply that we should only concern ourselves with the performance of a single stock; say $Y_1$.  This is readily seen as a consequence of known results on gambling with side information \cite[Theorem 6.2.1]{bib:CoverThomas2006}, since putting $b(X^n)=X_1$ yields
\begin{align}
I(b(X^n);Y^n) = I(X_1;Y^n) = I(X_1;Y_1)=1-H(\alpha),
\end{align}
hence the conjectured upper bound \eqref{eqn:conjUB} is attainable and represents the maximum possible increase in doubling rate.

Finally, we point out that \eqref{eqn:conjUB} is related in spirit to the notion of \emph{average sensitivity} of Boolean functions.  This topic has received a great deal of attention in the computer science literature (cf. \cite{bib:O'Donnell2008}).  To see the connection to sensitivity, note that \eqref{eqn:conjUB} can be rewritten as
\begin{align}
H(b(X^n)|Y^n) \geq H(b(X^n))-1+H(\alpha). \label{eqn:sensitivity}
\end{align}
For fixed $\Pr\{b(X^n)=0\}$, the right hand side of \eqref{eqn:sensitivity} is constant.  Hence, the conjecture essentially lower bounds the output uncertainty  of Boolean functions with respect to noisy inputs.

This paper is organized as follows.  Section \ref{sec:results} provides a summary of the main results and their implications. It includes a refinement  of Conjecture \ref{conj:Ib_Yn} by splitting it into two ``sub-conjectures."  The following section deals with the proofs of the main results. Section \ref{sec:conclusion} delivers  concluding remarks.

\section{Results and Implications}\label{sec:results}
Let $X^n$ be a sequence of i.i.d.~$\mbox{Bernoulli}\left(1/2\right)$ random variables, $Z^n$ be a sequence of i.i.d. $\mbox{Bernoulli}\left(\alpha \right)$ random variables independent of $X^n$, $0\leq\alpha\leq 1/2$. Let $Y^n = X^n \oplus Z^n$, where ``$\oplus$" denotes coordinate-wise XOR.  Throughout, we let $\Omega=\{0,1\}$, $\Omega_n=\{0,1\}^n$, and consider Boolean functions  $b:\Omega_n\rightarrow \Omega$.

\begin{definition}
The lexicographical ordering $\prec_L$ on $\{0,1\}^k$ is defined as follows:  $x^k\prec_L \tilde{x}^k$ iff $x_j < \tilde{x}_j$ for some $j$ and $x_i = \tilde{x}_i$ for all $i<j$.
\end{definition}
For example, if $k=3$, we have $000 \prec_L 001 \prec_L 010 \prec_L 011 \prec_L 100 \prec_L 101 \prec_L 110 \prec_L 111$.

\begin{definition}
We define $L_k(M)$ to be the initial segment of size $M$ in the lexicographical ordering on $\{0,1\}^k$.  For example, $L_3(4) = \{000,001,010,011\}$.
\end{definition}

For a function $b: {\Omega_n}\rightarrow \Omega$, we say that ``$b$ is \emph{lex}'' when $b^{-1}(0)=L_n(|b^{-1}(0)|)$.  In other words, $b$ is lex when it maps an initial segment of the lexicographical order to $0$, and the complement segment to $1$.

Instead of dealing with Conjecture \ref{conj:Ib_Yn} directly, consider the following two conjectures:

\begin{conjecture} \label{conj:comp}
For a given $n$ and fixed cardinality $|b^{-1}(0)|$, the conditional entropy $H(b(X^n)|Y^n)$ is minimized when $b$ is lex.
\end{conjecture}

\begin{conjecture}\label{conj:bLexLB}
If $b:\Omega_n\rightarrow \Omega$ is lex, then
\begin{align}
H(b(X^n)|Y^n) \geq H(b(X^n))H(\alpha).\label{eqn:bLexLB}
\end{align}
\end{conjecture}

Clearly,  Conjecture \ref{conj:Ib_Yn} would follow as a corollary if Conjectures \ref{conj:comp} and \ref{conj:bLexLB} were valid.

Referring to Conjecture \ref{conj:bLexLB} as a ``conjecture" is perhaps too modest.  Indeed, we derive a simple recursive algorithm capable of proving \eqref{eqn:bLexLB} for any fixed $\alpha$.  With the assistance of a computer, this algorithm has verified \eqref{eqn:bLexLB} for $\alpha$ ranging from $0$ to $1/2$ in increments of $0.001$. Refer to Theorem \ref{thm:algoVerify} and the following discussion in Section \ref{sec:pseudo} for details.

\subsection{Conjecture \ref{conj:comp} and Isoperimetry}
Conjecture \ref{conj:comp} is reminiscent of a classical theorem in discrete mathematics originally due to Harper \cite{bib:Harper} that gives an exact edge-isoperimetric inequality for the hypercube.  To state the theorem, we need a few basic notations.  Let $Q_n$ be the $n$-dimensional hypercube, and let $V(Q_n)=\Omega_n$ be its set of vertices.  For $S\subseteq V(Q_n)$, the edge boundary $\partial (S)$ is  the set of edges one has to delete to disconnect $S$ from any vertex not in $S$.
\begin{theorem}\label{thm:Harper}
For $S\subseteq V(Q_n)$ with $|S|=k$, we have $|\partial(S)|\geq |\partial(L_n(k))|$.
\end{theorem}
The simplest proofs of Theorem \ref{thm:Harper} rely on so-called \emph{compression operators}, popularized by Bollob\'{a}s and Leader \cite{bib:Compressions}.  These compression operators turn out to be useful in making progress towards Conjecture \ref{conj:comp}, so we introduce them now.

Let $\mathcal{I}$ be subset of $\{1,2,\dots,n\}$ of cardinality $k$.  To be concrete, let $\mathcal{I}=\{i_1,i_2,\dots,i_k\}$, where $i_1 < i_2 <\cdots < i_k$.  For a set $B\subseteq \Omega_n$ and $x^n$ having $x_i = 0$ for all $i\in \mathcal{I}$, we define the $\mathcal{I}$-section of $B$ at $x^n$ by
\begin{align}
B_{\mathcal{I}}(x^n) = \left\{z^k : y^n\in B,
y_i = \left\{\begin{array}{ll}
z_{j} & \mbox{if $i = i_j\in\mathcal{I}$}\\
x_i & \mbox{otherwise}
\end{array} \right. \right\}. \label{eqn:Isection}
\end{align}

For instance, if $B=\{000,001,011,101\}$, then examples of $\mathcal{I}$-sections at different $x^3\in\Omega_3$ are given by:
\begin{align}
B_{\{1\}}(001) &= \{0,1\},\\
B_{\{2\}}(100) &= \emptyset,\\
B_{\{1,2\}}(000) &= \{00\},\\
B_{\{1,2\}}(001) &= \{00,01,10\}.
\end{align}

The $\mathcal{I}$-compression of $B$, $C_{\mathcal{I}}(B)$, is defined in terms of its $\mathcal{I}$-sections
\begin{align*}
\left(C_{\mathcal{I}}(B)\right)_{\mathcal{I}}(x^n) = L_k\left( |B_{\mathcal{I}}(x^n)|\right).
\end{align*}
In other words, $C_{\mathcal{I}}$ replaces each $\mathcal{I}$-section of $B$ with an initial segment of the lexicographical order.  We say that $B$ is $\mathcal{I}$-compressed if $C_{\mathcal{I}}(B)= B$.  Note that $C_{\mathcal{I}}(B)$ is always $\mathcal{I}$-compressed.

Continuing the above example of
$B=\{000,001,011,101\}$,  example $\mathcal{I}$-compressions are given by:
\begin{align}
C_{\{1\}}(B) &= \{000,001,011,101\},\\
C_{\{2,3\}}(B) &= \{000, 001, 010, 100\}.
\end{align}

We pause to make two important observations.  First,
$\mathcal{I}$-compression preserves the size of the set on which it operates.  That is, $|C_{\mathcal{I}}(B)|= |B|$.  Second, if $B$ is $\mathcal{I}$-compressed, then it is also $\mathcal{J}$-compressed for all $\mathcal{J}\subset \mathcal{I}$.

The following theorem states that when $|\mathcal{I}|=2$, applying an $\mathcal{I}$-compression to  $b^{-1}(0)$ does not decrease the information $b(X^n)$ reveals about $Y^n$.  Thus, compression provides a method of modifying functions in a manner that does not adversely affect the mutual information $I(b(X^n);Y^n)$.

\begin{theorem}\label{thm:2CompressionHelps}
Let $b:\Omega_n\rightarrow \Omega$ and let $\mathcal{I}\subset\{1,\dots,n\}$ satisfy $|\mathcal{I}|=2$. If $\hat{b}:\Omega_n\rightarrow \Omega$ is defined by its preimage $\hat{b}^{-1}(0) = C_{\mathcal{I}}({b}^{-1}(0))$, then $I(\hat{b}(X^n);Y^n)\geq I({b}(X^n);Y^n)$.
\end{theorem}

By definition, if $C_{\mathcal{I}}(\cdot)$ changes an element of $b^{-1}(0)$, it moves it \emph{lower} in the lexicographical ordering on $\Omega_n$.  Therefore, one can repeatedly apply Theorem \ref{thm:2CompressionHelps} for different subsets $\mathcal{I}$ of cardinality 2, ultimately terminating at a function $\hat{b}$ which is $\mathcal{I}$-compressed for all $\mathcal{I}$ with $|\mathcal{I}|\leq 2$. Hence, we have the following corollary.

\begin{corollary}\label{cor:compSuff}  Let $\mathcal{S}_n$ be the set of functions $b:\Omega_n\rightarrow \Omega$ for which $b^{-1}(0)$ is $\mathcal{I}$-compressed for all $\mathcal{I}$ with $|\mathcal{I}|\leq 2$.
In maximizing $I(b(X^n);Y^n)$, it is sufficient to consider functions $b \in \mathcal{S}_n$.
\end{corollary}

The implications of Theorem \ref{thm:2CompressionHelps} and its corollary are twofold.
First, it allows the  verification of Conjecture \ref{conj:comp} for modest values of $n$.  Indeed, we have numerically validated Conjectures \ref{conj:Ib_Yn} and  \ref{conj:comp} for $n\leq 7$ by evaluating $I(b(X^n);Y^n)$ for $b \in \mathcal{S}_n$.  To appreciate the reduction afforded by Corollary \ref{cor:compSuff}, define $\mathcal{B}_n$ to be the set of all $2^{2^n}$ Boolean functions on $n$ inputs.  A comparison between $|\mathcal{S}_n|$ and $|\mathcal{B}_n|$ is given in Table \ref{tbl:cost}.
\vskip1ex

\begin{table}
{\normalsize
\begin{center}
\begin{tabular}{ c || c | c }
$n$ & $ |\mathcal{S}_n|$ & $|\mathcal{B}_n|$ \\
  \hline
  2 & 5 & 16 \\
  3 & 10 & 256 \\
    4 & 25 & 65,536 \\
      5 & 119 & $4.3\times 10^9$ \\
        6 & 1173 & $1.8\times 10^{19}$ \\
          7 & 44,315 & $3.4\times 10^{38}$
\end{tabular}
\end{center}
}
\caption{Reduction in number of candidate Boolean functions to be considered for verification of Conjecture \ref{conj:comp}.}
\label{tbl:cost}
\end{table}
\vskip1ex
Second, Theorem \ref{thm:2CompressionHelps} reinforces the intuition behind Conjecture \ref{conj:comp}.  As we noted above, if $C_{\mathcal{I}}(\cdot)$ changes an element of $b^{-1}(0)$, it moves it \emph{lower} in the lexicographical ordering on $\Omega_n$.  Thus, roughly speaking, applying $\mathcal{I}$-compression to $b^{-1}(0)$ yields a function $\hat{b}$ which is (i) closer to an initial segment of the lexicographical order, and (ii) for $|\mathcal{I}|\leq 2$ satisfies $H(\hat{b}(X^n)|Y^n)\leq H({b}(X^n)|Y^n)$.

Ideally, Theorem \ref{thm:2CompressionHelps}  should generalize to include $\mathcal{I}$-compressions for $|\mathcal{I}|>2$.  Indeed, if we could take $|\mathcal{I}|=n$, Conjecture \ref{conj:comp} would be proved.  However, we have found counterexamples where compression \emph{increases} $H(b(X^n)|Y^n)$ for $|\mathcal{I}|>2$ (but still reduces $H(b(X^n)|Y^n)$ for $|\mathcal{I}|=n$).  We omit the details.%

\subsection{An Algorithmic Proof of Conjecture \ref{conj:bLexLB}}\label{sec:pseudo}
Now, we turn toward establishing Conjecture \ref{conj:bLexLB}.  Unless otherwise specified, all Boolean functions in this subsection are assumed to be lex.

Define $f(x)=-x\log x$.  Note that if $b$ is lex, then so is $\neg b$ (i.e., the negation of $b$) up to a relabeling of inputs.  Therefore, to prove \eqref{eqn:bLexLB}, it is sufficient to prove
\begin{align}
\EE_{Y^n} &f\left(\Pr\{b(X^n)=0|Y^n\}\right) \notag\\
&\geq f\left(\Pr\{b(X^n)=0\}\right) H(\alpha). \label{eqn:fIneq}
\end{align}

To simplify notation, for a dyadic rational $p=k/2^n$, define
\begin{align}
T_{\alpha}(p) \triangleq \EE_{Y^n} f\left(\Pr\{b(X^n)=0|Y^n\}\right),
\end{align}
where $b$ is the unique lex function on $n$ inputs with $\Pr\{b(X^n)=0\} = p$.  Note that if $k$ is even, $b$ does not depend on its input bit $x_n$.  Therefore, $T_{\alpha}(p)$ is well-defined for all dyadic rationals $p\in [0,1]$. It is a simple exercise to show that $T_{\alpha}(\cdot)$ is continuous on the dyadic rationals (in fact, it is H\"{o}lder continuous with exponent 1/2). Therefore, $T_{\alpha}(p)$ is also well-defined when $p\in[0,1]$ is not a dyadic rational, by considering its unique continuous extension to $[0,1]$.

Thus, the validity of \eqref{eqn:fIneq} for all lex $b$ (and all $n$) is equivalent to the inequality
\begin{align}
T_{\alpha}(p) &\geq f\left(p\right) H(\alpha) \quad \forall p\in[0,1], \label{eqn:fIneqT}
\end{align}
motivating the following theorem.

\begin{theorem}\label{thm:algoVerify}
Fix $\alpha\in (0,1/2)$. If a call to Algorithm \ref{alg:testIneq} with arguments $(p_{-}, p_{+})=(1/2,1)$ eventually terminates, then Conjecture \ref{conj:bLexLB} is true for the chosen $\alpha$.
\end{theorem}
\vspace{-15pt}

\begin{center}
\begin{pseudocode}[framebox]{TestInequality}{p_{-}, p_{+}}\label{alg:testIneq}

\MAIN

\IF \CALL{CheckChord}{p_-,p_+} < 0
\THEN
\BEGIN
p \GETS \frac{1}{2}(p_- + p_+)\\
\CALL{TestInequality}{p_-, p} \\
\CALL{TestInequality}{p, p_+}
\END
\ENDMAIN\\
~\\

\PROCEDURE{CheckChord}{a,b}
\COMMENT{$\begin{array}{ll}
C(x) \mbox{~is  the chord connecting} \\ \mbox{the points~} (a,T_{\alpha}(a)) \mbox{~and~} (b,T_{\alpha}(b)).
\end{array}$}\\
C(x) := \frac{T_{\alpha}(b)-T_{\alpha}(a)}{b-a}(x-a)+T_{\alpha}(a)\\
\nu \GETS \min_{x\in[a,b]}  C(x)- f(x)H(\alpha)\\

\RETURN{\nu}
\vspace{-18pt}
\ENDPROCEDURE
\end{pseudocode}
\end{center}
\vspace{-12pt}

\begin{remark} In the subroutine {\sc CheckChord}($a,b$) of Algorithm \ref{alg:testIneq}, the minimization   %
has a closed form solution.
\end{remark}

\begin{figure}
  \centering
      \includegraphics[width=0.47\textwidth,clip=true, trim=18mm 5mm 10mm 0mm]{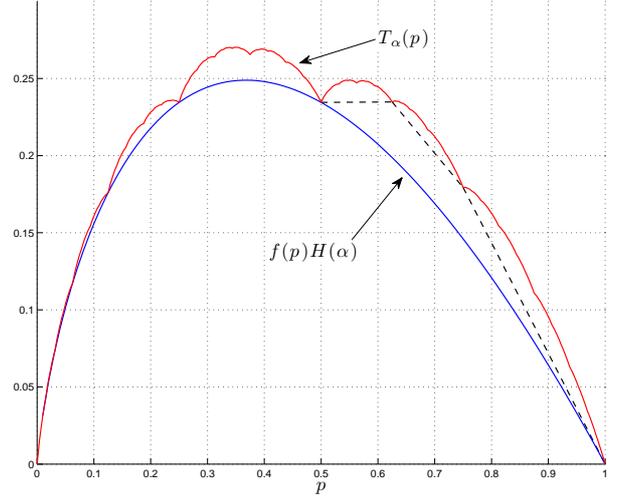}
  \caption{A comparison of $T_{\alpha}(p)$ and $f(p)H(\alpha)$ for $\alpha=0.1$. The broken line shows the three chords Algorithm  \ref{alg:testIneq} constructs before terminating.}\label{fig:takagi}\vspace{-13pt}
\end{figure}

In words, Algorithm \ref{alg:testIneq} recursively constructs a piecewise linear function on the interval $p\in[1/2,1]$ which simultaneously upper bounds $f(p)H(\alpha)$ and lower bounds $T_{\alpha}(p)$.  As discussed in Section \ref{sec:pseudoProof}, this is sufficient to prove \eqref{eqn:fIneqT}. Figure \ref{fig:takagi} illustrates this procedure
for $\alpha=0.1$.

Using a Matlab implementation of Algorithm \ref{alg:testIneq}, we have validated \eqref{eqn:fIneqT} for $\alpha$ ranging from $0$ to $1/2$ in increments of $0.001$.  Hence, it is reasonable to believe that Conjecture \ref{conj:bLexLB} is true in general.

Despite the apparent gap between $T_{\alpha}(p)$ and $f(p)H(\alpha)$ for $p\in(1/2,1)$ (e.g., Fig.~\ref{fig:takagi}), the oscillatory behavior of $T_{\alpha}(p)$ seems to render traditional analysis techniques ineffective in establishing \eqref{eqn:fIneqT}.  This was our motivation for pursuing an algorithmic proof.  To get a sense for the strange behavior of $T_{\alpha}(p)$, we point out that it is possible to show that $\lim_{\alpha\rightarrow 0}T_{\alpha}(p)/H(\alpha)$ is equal to the \emph{Takagi function}, a classical construction of an everywhere-continuous, nowhere-differentiable function closely related to the edge-isoperimetric inequality given in Theorem \ref{thm:Harper} (cf. \cite{bib:Allaart, bib:mcFunction}).   We omit the details due to space constraints.

\section{Proofs}\label{sec:proof}

\subsection{Proof of Theorem \ref{thm:2CompressionHelps}}
We begin the proof of Theorem \ref{thm:2CompressionHelps} by first proving the following result for 1-dimensional compressions.  
\begin{lemma}\label{lem:1CompressionHelps}
Let $b:\Omega_n\rightarrow \Omega$ and $i\in\{1,2,\dots,n\}$. If $\hat{b}:\Omega_n\rightarrow \Omega$ is defined by its preimage $\hat{b}^{-1}(0) = C_{\{i\}}({b}^{-1}(0))$, then $I(\hat{b}(X^n);Y^n)\geq I({b}(X^n);Y^n)$.
\end{lemma}

\begin{proof}
It suffices to consider the case where $i=n$, as any other case can be handled by first permuting coordinates.

 \if 0
By possibly inverting the $n^{th}$ input of $b(\cdot)$, we can assume without loss of generality that
\begin{align*}
\Pr\{b(X^n)=0|y^{n-1},0\} \geq \Pr\{b(X^n)=0|y^{n-1},1\}
\end{align*}
  for a fixed $y^{n-1}$.

 To see this, define $A_0=\{x^{n-1} : (x^{n-1},0) \in b^{-1}(0)\}$ and $A_1=\{x^{n-1} : (x^{n-1},1) \in b^{-1}(0)\}$.  By considering the function $b$, or equivalently the function
 %
$b'(x^n) := b(x^{n-1},\neg x_n)$, we can guarantee that
\begin{align}
\sum_{x^{n-1}\in A_0} p(x^{n-1}|y^{n-1}) \geq \sum_{x^{n-1}\in A_1} p(x^{n-1}|y^{n-1}).
\end{align}
Then note that, since $\alpha\leq 1/2$, and
\begin{align*}
\Pr\{b(X^n)=0|y^{n-1},0\} &= (1-\alpha)\sum_{x^{n-1}\in A_0} p(x^{n-1}|y^{n-1})  \\
&+ \alpha \sum_{x^{n-1}\in A_1} p(x^{n-1}|y^{n-1})\\
\Pr\{b(X^n)=0|y^{n-1},1\} &= \alpha \sum_{x^{n-1}\in A_0} p(x^{n-1}|y^{n-1})  \\
&+ (1-\alpha) \sum_{x^{n-1}\in A_1} p(x^{n-1}|y^{n-1}),
\end{align*}
we have $\Pr\{b(X^n)=0|y^{n-1},0\} \geq \Pr\{b(X^n)=0|y^{n-1},1\}$ as desired.
\fi

Define $B=b^{-1}(0)$ and $\hat{b}^{-1}(0) = C_{\{n\}}(B)$, and let
\begin{align}
E_0 &= \left\{x^{n-1} : B_{\{n\}}(x^{n-1},0)=\{0\}\right\}\\
E_1 &= \left\{x^{n-1} :B_{\{n\}}(x^{n-1},0)=\{1\}\right\},
\end{align}
where $B_{\{n\}}$ is defined by \eqref{eqn:Isection} with $\mathcal{I}=\{n\}$.
Define $\breve{b}(x^{n-1},x_n):= \hat{b}(x^{n-1},\neg x_n)$, where $\neg x_n$ is the negation of $x_n$.

By definition of $\hat{b}$ and $\breve{b}$, we have the identities
\begin{align*}
&\Pr\{{b}(X^n)=0|y^{n-1},0\} \\
&= \Pr\{\hat{b}(X^n)=0|y^{n-1},0\} -(1-2\alpha)\Pr\{X^{n-1}\in E_1 | y^{n-1}\} \\
&=\Pr\{\breve{b}(X^n)=0|y^{n-1}, 0\} +(1-2\alpha)\Pr\{X^{n-1}\in E_0 | y^{n-1}\}.
\end{align*}
Similar identities hold for $\Pr\{{b}(X^n)=0|y^{n-1},1\}$ with opposite signs on the $(1-2\alpha)$ terms, giving 
\begin{align}
&\Pr\{{b}(X^n)=0|y^{n}\} =\notag \\
&\quad \theta\Pr\{\hat{b}(X^n)=0|y^{n}\} + (1-\theta)\Pr\{\breve{b}(X^n)=0|y^{n}\} ,
\end{align}
where
\begin{align}
\theta = \frac{\Pr\{X^{n-1}\in E_0 | y^{n-1}\}}{\Pr\{X^{n-1}\in E_0 | y^{n-1}\}+\Pr\{X^{n-1}\in E_1 | y^{n-1}\}}.
\end{align}
Concavity of entropy implies that
\begin{align*}
\theta H(\hat{b}(X^n)|y^n) + (1-\theta ) H(\breve{b}(X^n)|y^n) \leq H({b}(X^n)|y^n).%
\end{align*}
Noting that $\theta$ only depends on $y^{n-1}$, we average both sides over $y_n\in\{0,1\}$ to obtain
\begin{align}
&\theta H(\hat{b}(X^n)|y^{n-1},Y_n) + (1-\theta ) H(\breve{b}(X^n)|y^{n-1},Y_n) \notag\\
&\quad\leq H({b}(X^n)|y^{n-1},Y_n). \label{cvxComb1}
\end{align}
By symmetry, $H(\hat{b}(X^n)|y^{n-1},Y_n)=H(\breve{b}(X^n)|y^{n-1},Y_n)$. Therefore,  averaging \eqref{cvxComb1} over all values of $y^{n-1}$, we can conclude that $H(\hat{b}(X^n)|Y^n) \leq H({b}(X^n)|Y^n)$. To complete the proof, we recall that $|\hat{b}^{-1}(0)| = |C_{\{n\}}(b^{-1}(0))|= |b^{-1}(0)|$.  Combined with the fact that $X^n$ is uniformly distributed on $\Omega_n$, this implies that $H(\hat{b}(X^n))=H({b}(X^n))$, as desired.
\end{proof}

We are now in a position to finish the proof of Theorem {\ref{thm:2CompressionHelps}}, which is similar to the proof of Lemma \ref{lem:1CompressionHelps}.

\begin{proof}[Proof of Theorem  {\ref{thm:2CompressionHelps}}]
We assume that $\mathcal{I}=\{n-1,n\}$, as all other cases follow by a permutation of coordinates.  To simplify notation, we write $B=b^{-1}(0)$.

By a repeated application of Lemma \ref{lem:1CompressionHelps}, we can assume that $B$ is $\{n-1\}$- and $\{n\}$-compressed.  Thus, the $\mathcal{I}$-sections $B_{\mathcal{I}}(x^n)$ can only be one of the following: $\emptyset$, $\{00\}$, $\{00,01\}$, $\{00,10\}$, $\{00,01, 10\}$, or $\{00,01,10,11\}$.  Note that all of these sets are initial segments of the lexicographical order on $\Omega_2$ except $\{00,10\}$.  Hence, we aim to transform $B$ so that $B_{\mathcal{I}}(x^n)\neq \{00,10\}$.  To this end, define
\begin{align}
G_0 &= \left\{x^{n} : x_{n-1}=x_n=0, B_{\mathcal{I}}(x^{n})=\{00,01\}\right\}\\
G_1 &= \left\{x^{n} : x_{n-1}=x_n=0, B_{\mathcal{I}}(x^{n})=\{00,10\}\right\}.
\end{align}
Now, define $\hat{b}$ by  $\hat{b}^{-1}(0) = C_{\mathcal{I}}(B)$ and the function $\breve{b}$ by permuting the last two coordinates:
\begin{align}
\breve{b}(x^{n-2},x_{n},x_{n-1}) = \hat{b}(x^{n-2},x_{n-1},x_n).\label{eqn:permSym}
\end{align}
It is relatively straightforward to show that
\begin{align}
&\Pr\{{b}(X^n)=0|y^{n}\} =\notag \\
&\quad \theta\Pr\{\hat{b}(X^n)=0|y^{n}\} + (1-\theta)\Pr\{\breve{b}(X^n)=0|y^{n}\}  ,
\end{align}
where
\begin{align}
\theta = \frac{\Pr\{X^{n-2}\in G_0 | y^{n-2}\}}{\Pr\{X^{n-2}\in G_0 | y^{n-2}\}+\Pr\{X^{n-2}\in G_1 | y^{n-2}\}}.
\end{align}
Concavity of entropy implies that
\begin{align*}
\theta H(\hat{b}(X^n)|y^n) + (1-\theta ) H(\breve{b}(X^n)|y^n) \leq H({b}(X^n)|y^n).%
\end{align*}
Noting that $\theta$ only depends on $y^{n-2}$, we average both sides over $y_{n-1},y_n$ to obtain
\begin{align}
&\theta H(\hat{b}(X^n)|y^{n-2},Y_{n-1}^n) + (1-\theta ) H(\breve{b}(X^n)|y^{n-2},Y_{n-1}^n) \notag\\
&\quad\leq H({b}(X^n)|y^{n-2},Y_{n-1}^n).\label{cvxComb}
\end{align}
Crucially, the symmetry \eqref{eqn:permSym} implies that
\begin{align}
H(\hat{b}(X^n)|y^{n-2},Y_{n-1}^n) = H(\breve{b}(X^n)|y^{n-2},Y_{n-1}^n).
\end{align}
Combining this with \eqref{cvxComb} and averaging over $y^{n-2}$ proves $H(\hat{b}(X^n)|Y^n)\leq H({b}(X^n)|Y^n)$.  Since $H(\hat{b}(X^n)) = H({b}(X^n))$, the proof is complete.
\end{proof}

\subsection{Proof of Theorem \ref{thm:algoVerify}}\label{sec:pseudoProof}

The proof of Theorem \ref{thm:algoVerify} requires the following lemmas.

\begin{lemma}\label{lem:functional}
In order to prove \eqref{eqn:fIneqT}, it is sufficient to consider $p\in[1/2,1]$.
\end{lemma}
\begin{proof}
Suppose $b$ and $b'$ are both lex and satisfy\footnote{Any lex function with $\Pr\{b(X^{n})=0\}=k/2^n$ can be reduced to a lex function on $n-1$ inputs if $k$ is even.}
\begin{align}
p \triangleq \Pr\{b'(X^n)=0\} = \frac{1}{2}\Pr\{b(X^{n-1})=0\}.
\end{align}
We have the identities $f(pq)=pf(q)+qf(p)$ and 
\begin{align}
\Pr\{b'(X^n)=0|y^n\} = \Pr\{X_1=0|y_1\}\Pr\{b(X_2^n)=0|y_2^n\},\notag
\end{align}
which imply %
the relation: 
$2T_{\alpha}(p) = T_{\alpha}(2p) + 2 p H(\alpha)$.
It follows easily that
\begin{align}
&\frac{1}{2}\Big[ T_{\alpha}(2p)-f(2p)H(\alpha)\Big] = \Big[ T_{\alpha}(p)-f(p)H(\alpha)\Big].
\end{align}
Thus, the claim is proved.
\end{proof}
Although $T_{\alpha}(p)$ is not concave, we are able to prove a pseudo-concavity characteristic of $T_{\alpha}(p)$.  This is exploited in the following claim.
\begin{lemma}\label{lem:pseudoConcave}
For $k\leq2^n$, consider the lex functions $b_{-}$ and $b_{+}$ which satisfy
\begin{align}
p_{-} \triangleq \Pr\{b_{-}(X^n)=0\}&=\frac{k}{2^n}  \label{eqn:pm}\\
p_{+} \triangleq \Pr\{b_{+}(X^n)=0\}&=\frac{k+1}{2^n} .\label{eqn:pp}
\end{align}
For $\theta\in[0,1]$, the following  inequality holds:
\begin{align*}
&T_{\alpha}\left(\theta p_- + (1-\theta)p_+\right)\geq \theta T_{\alpha}(p_-) + (1-\theta)T_{\alpha}(p_+).
\end{align*}
\end{lemma}
\begin{proof}
First, observe that it suffices to prove
\begin{align}
&T_{\alpha}\left(\frac{p_-+p_+}{2}\right)\geq \frac{1}{2}\Big[ T_{\alpha}(p_-) + T_{\alpha}(p_+)\Big].  \label{eqn:desiredStep}
\end{align}
Indeed, from \eqref{eqn:desiredStep} an inductive argument proves the lemma when $\theta p_- + (1-\theta)p_+$ is restricted to the set of dyadic rationals.  Then, recalling the continuity of $T_{\alpha}(\cdot)$ on $[0,1]$ completes the proof.

To this end, let $b$ be the unique lex function on $n+1$ inputs which satisfies
\begin{align}
&\Pr\{b(X^{n+1})=0\}=\frac{2k+1}{2^{n+1}}=\frac{1}{2}\Big[p_- +p_+ \Big]. \label{p+p-}
\end{align}
By construction, we have
\begin{align}
&\Pr\{b(X^{n+1})=0|Y^{n+1}\} \notag\\
&~~= \Pr\{X_{n+1}=0|Y_{n+1}\}\Pr\{b_{+}(X_1^{n})=0|Y_1^{n}\} \notag\\
&~~~~~+\Pr\{X_{n+1}=1|Y_{n+1}\}\Pr\{b_{-}(X_1^{n})=0|Y_1^{n}\}.\label{eqn:btopbot}
\end{align}
Combining  \eqref{p+p-} and \eqref{eqn:btopbot} with the fact that $f(x)$ is concave, we have the desired inequality
\begin{align}
&T_{\alpha}\left(\frac{p_-+p_+}{2}\right)\notag\\
&\geq \EE_{Y^{n+1}} \Pr\{X_{n+1}=0|Y_{n+1}\}f \left(\Pr\{b_{+}(X_1^{n})=0|Y_1^{n}\}\right)\notag\\
&~~+\EE_{Y^{n+1}} \Pr\{X_{n+1}=1|Y_{n+1}\}f \left(\Pr\{b_{-}(X_1^{n})=0|Y_1^{n}\}\right)\notag\\
&=\frac{1}{2}\Big[ T_{\alpha}(p_-) + T_{\alpha}(p_+)\Big].
\end{align}
\end{proof}

We are now in a position to prove Theorem \ref{thm:algoVerify}.

\begin{proof}[Proof of Theorem \ref{thm:algoVerify}]
Lemmas \ref{lem:functional} and \ref{lem:pseudoConcave} imply that, in order to prove \eqref{eqn:fIneqT},  it is sufficient to construct a piecewise linear function $g:[1/2,1]\rightarrow [0,\infty)$ satisfying the following properties:
\begin{enumerate}
\item Each segment of $g$ is a chord connecting the points $(p_-,T_{\alpha}(p_-))$ and  $(p_+,T_{\alpha}(p_+))$, where $p_-$ and $p_+$ are of the form %
\begin{align}
&p_{-}=\frac{k}{2^n},
&p_{+} =\frac{k+1}{2^n}.
\end{align}
for some  integers $k,n$.
\item For $p\in[1/2,1]$, $g(p)\geq f(p) H(\alpha)$.
\end{enumerate}
By definition, Algorithm \ref{alg:testIneq} terminates only if it constructs such a function.
\end{proof}

\section{Concluding Remarks} \label{sec:conclusion}
Although Conjecture \ref{conj:Ib_Yn} remains open, we have provided substantial evidence in support of its validity.  Indeed, our results suggest that Conjecture \ref{conj:comp} is valid and we have an algorithmic proof establishing Conjecture \ref{conj:bLexLB} for any given value of $\alpha$.  Any complete proof of Conjectures \ref{conj:Ib_Yn} or  \ref{conj:comp} would be of significant interest, since it would likely require new methods which may be applicable in information theory and elsewhere (e.g., in proving discrete isoperimetric inequalities).

We leave the reader with a  weak form of Conjecture \ref{conj:Ib_Yn} which could provide insight.  For Boolean functions $b,b'$, does it hold that $I(b(X^n);b'(Y^n))\leq 1-H(\alpha)$?  ~While this problem appears difficult in general, it is a simple exercise to show this is true when $b(X^n)$ and $b'(Y^n)$ are both Bernoulli(1/2).  Intuitively, this should be the case for $b,b'$ which maximize $I(b(X^n);b'(Y^n))$.

\section*{Acknowledgment}
The authors are grateful to Abbas El Gamal, Chandra Nair, and Yeow-Khiang Chia for many helpful discussions.

This work is supported in part by the Air Force grant FA9550-10-1-0124 and by the NSF Center for Science of Information
under grant agreement CCF-0939370.

\bibliographystyle{IEEEtran.bst}
\bibliography{myBib}

\begin{thebibliography}{1}
\providecommand{\url}[1]{#1}
\csname url@samestyle\endcsname
\providecommand{\newblock}{\relax}
\providecommand{\bibinfo}[2]{#2}
\providecommand{\BIBentrySTDinterwordspacing}{\spaceskip=0pt\relax}
\providecommand{\BIBentryALTinterwordstretchfactor}{4}
\providecommand{\BIBentryALTinterwordspacing}{\spaceskip=\fontdimen2\font plus
\BIBentryALTinterwordstretchfactor\fontdimen3\font minus
  \fontdimen4\font\relax}
\providecommand{\BIBforeignlanguage}[2]{{%
\expandafter\ifx\csname l@#1\endcsname\relax
\typeout{** WARNING: IEEEtran.bst: No hyphenation pattern has been}%
\typeout{** loaded for the language `#1'. Using the pattern for}%
\typeout{** the default language instead.}%
\else
\language=\csname l@#1\endcsname
\fi
#2}}
\providecommand{\BIBdecl}{\relax}
\BIBdecl

\bibitem{bib:Klotz2012}
J.~G. Klotz, D.~Kracht, M.~Bossert, and S.~Schober, ``Canalizing boolean
  functions maximize the mutual information,'' \emph{arxiv:1207.7193}, 2012.

\bibitem{bib:Samal2008}
A.~Samal and S.~Jain, ``The regulatory network of e. coli metabolism as a
  boolean dynamical system exhibits both homeostasis and flexibility of
  response,'' \emph{BMC Systems Biology}, vol.~2, no.~1, p.~21, 2008.

\bibitem{bib:Tishby1999}
N.~Tishby, F.~C. Pereira, and W.~Bialek, ``The information bottleneck method,''
  in \emph{The 37th Annual Allerton Conference on Communication, Control, and
  Computing}, September 1999, pp. 368 -- 377.

\bibitem{bib:CoverThomas2006}
T.~M. Cover and J.~A. Thomas, \emph{Elements of Information Theory},
  2nd~ed.\hskip 1em plus 0.5em minus 0.4em\relax John Wiley \& Sons, 2006.

\bibitem{bib:O'Donnell2008}
R.~O'Donnell, ``Some topics in analysis of boolean functions,'' in \emph{Proc.
  STOC '08}.\hskip 1em plus 0.5em minus 0.4em\relax New York, NY, USA: ACM,
  2008, pp. 569--578.

\bibitem{bib:Harper}
L.~H. Harper, ``Optimal numberings and isoperimetric problems on graphs,,''
  \emph{Journal of Combinatorial Theory}, no.~1, pp. 385 -- 393, 1966.

\bibitem{bib:Compressions}
B.~Bollob\'{a}s and I.~Leader, ``Compressions and isoperimetric inequalities,''
  \emph{J. Combinatorial Theory, Series A}, vol.~56, no.~1, pp. 47 -- 62, 1991.

\bibitem{bib:Allaart}
P.~C. Allaart and K.~Kawamura, ``The {T}akagi function: a survey,'' \emph{Real
  Analysis Exchange}, vol.~37, no.~1, pp. 1 -- 54, 2011.

\bibitem{bib:mcFunction}
C.~J. Guu, ``The mcfunction,'' \emph{Discrete Mathematics}, vol. 213, no. 1Ð3,
  pp. 163 -- 167, 2000.

\end{thebibliography}

\end{document}